\documentclass[11pt]{article}

\usepackage{amssymb,amsmath}
\usepackage{amsthm}
\usepackage{latexsym}

\newcount\hour \newcount\minutes \hour=\time \divide\hour by 60
\minutes=\hour \multiply\minutes by -60 \advance\minutes by \time
\def\mmmddyyyy{\ifcase\month\or Jan\or Feb\or Mar\or Apr\or May\or
Jun\or Jul\or Aug\or Sep\or Oct\or Nov\or Dec\fi \space\number\day,
\number\year}
\def\hhmm{\ifnum\hour<10 0\fi\number\hour :%
  \ifnum\minutes<10 0\fi\number\minutes} 
\makeatother

\begin{document}

\title{Dichotomy Results for Fixed Point Counting in Boolean Dynamical Systems\thanks{A preliminary version of this paper \cite{homan-kosub-2007:proc} was presented
at the 10th Italian Conference on Theoretical Computer Science (ICTCS'07).}}

\author{%
{\em Christopher M. Homan}\\
Department of Computer Science\\
Rochester Institute of Technology\\
Rochester, NY 14623, USA\\
\texttt{cmh@cs.rit.edu}
\and
{\em Sven Kosub} \\
Department of Computer \& Information Science\\
University of Konstanz\\
Box 67, D-78457 Konstanz, Germany\\
\texttt{Sven.Kosub@uni-konstanz.de}
}

\renewcommand{\P}{\ensuremath{ {\rm P}}}
\newcommand{{\F}}{{\cal F}}

\newcommand{\co}{\ensuremath{ {\rm co}}}
\newcommand{\NP}{\ensuremath{ {\rm NP}}}
\newcommand{\FPT}{\ensuremath{ {\rm FPT}}}
\newcommand{\W}{\ensuremath{ {\rm W}}}
\newcommand{\PSPACE}{\ensuremath{ {\rm PSPACE}}}
\newcommand{\defin}{\in}

\newtheorem{definition}{Definition}
\newtheorem{theorem}[definition]{Theorem}
\newtheorem{corollary}[definition]{Corollary}
\newtheorem{lemma}[definition]{Lemma}
\newtheorem{proposition}[definition]{Proposition}
\newtheorem{conjecture}[definition]{Conjecture}
\newtheorem{remark}[definition]{Remark}
\newtheorem{example}[definition]{Example}
\newtheorem{observation}[definition]{Observation}

\newcommand{{\Nat}}{{\rm I\! N}}
\newcommand{{\Real}}{{\rm I\! R}}
\newcommand{{\BF}}{{\rm BF}}
\newcommand{{\Forb}}{{\rm Forb}}
\newcommand{{\AllGraphs}}{{\rm AllGraphs}}
\newcommand{{\pot}}{{\cal P}}
\newcommand{{\D}}{{\cal D}}
\newcommand{{\BFD}}{{\rm D}}
\newcommand{{\BFR}}{{\rm R}}
\newcommand{{\BFM}}{{\rm M}}
\newcommand{{\BFN}}{{\rm N}}
\newcommand{{\BFL}}{{\rm L}}
\newcommand{{\BFE}}{{\rm E}}
\newcommand{{\BFV}}{{\rm V}}
\newcommand{{\BFS}}{{\rm S}}
\newcommand{{\pow}}{{\cal P}}
\newcommand{{\XA}}{{\cal X}}
\newcommand{{\XB}}{{\cal Y}}
\newcommand{{\XC}}{{\cal C}}
\newcommand{{\XD}}{{\cal D}}
\newcommand{{\XS}}{{\cal S}}
\newcommand{{\XG}}{{\cal G}}
\newcommand{{\XU}}{{\cal U}}
\newcommand{{\XI}}{{\cal I}}
\newcommand{{\XX}}{{\cal X}}
\newcommand{{\XV}}{{\cal V}}
\newcommand{{\XF}}{{\bf F}}

\newcommand{{\Proj}}{{\mathrm{proj}}}
\newcommand{{\Exp}}{{\mathrm{exp}}}
\newcommand{{\df}}{{\mathrm{def}}}
\newcommand{{\nt}}{{\mathrm{not}}}
\newcommand{{\tw}}{{\mathrm{tw}}}
\newcommand{{\id}}{{\mathrm{id}}}
\newcommand{{\argmax}}{{\mathrm{arg max}}}
\newcommand{{\rank}}{{\mathrm{rank}}}
\newcommand{{\sd}}{{\mathrm{sd}}}
\newcommand{{\et}}{{\mathrm{et}}}
\newcommand{{\vel}}{{\mathrm{vel}}}
\newcommand{{\dual}}{{\mathrm{dual}}}
\newcommand{{\CSP}}{{\mathrm{CSP}}}
\newcommand{{\lookup}}{{\mathrm{T}}}
\newcommand{{\formula}}{{\mathrm{F}}}
\newcommand{{\circuit}}{{\mathrm{C}}}

\newcommand{{\FPE}}{{\textsc{FixedPoints}}}
\newcommand{{\SAT}}{{\textsc{3SAT}}}
\newcommand{{\PSAT}}{{\textsc{Planar 3SAT}}}
\newcommand{{\PHSAT}}{{\textsc{Planar Horn-2SAT}}}
\newcommand{{\POSSAT}}{{\textsc{Pos 2SAT}}}

\date{\empty}

\maketitle

\begin{abstract}
We present dichotomy theorems regarding the computational complexity of counting
fixed points in boolean (discrete) dynamical systems, i.e., finite discrete 
dynamical systems over the domain $\{0,1\}$.
For a class $\F$ of boolean functions and a class $\XG$ of graphs, an $(\F,
\XG)$-system is a boolean dynamical system with local transitions functions 
lying in $\F$ and graphs in $\XG$. 
We show that, if local transition functions are given by lookup tables, then
the following complexity classification holds:
Let $\F$ be a class of boolean functions closed under superposition and let 
$\XG$ be a graph class closed under taking minors. 
If $\F$ contains all $\min$-functions, all $\max$-functions, or all self-dual 
and monotone functions, and $\XG$ contains all planar graphs, then it is 
$\#\P$-complete to compute the number of fixed points in an 
$(\F,\XG)$-system; otherwise it is computable in polynomial time. 
We also prove a dichotomy theorem for the case that local 
transition functions are given by formulas (over logical bases).
This theorem has a significantly
more complicated structure than the theorem for lookup tables.
A corresponding theorem 
for boolean circuits coincides with the theorem for formulas. 

\end{abstract}


\section{Introduction}

Efforts to understand the behavior of complex systems have led to various models 
for finite discrete dynamical systems, including 
(finite) cellular automata (see, e.g., \cite{neumann-1966,wolfram-1994}), 
discrete recurrent Hopfield networks (see, e.g., \cite{hopfield-1982,bar-yam-2003}), and 
concurrent and communicating finite state machines (see, e.g., \cite{milner-1999,rabinovich-1997}).
A fairly general class of systems was introduced in \cite{barrett-mortveit-reidys-2000}. There, 
a finite discrete dynamical system (over a finite domain $\XD$) is defined as: 
(a) a finite undirected graph, where vertices correspond to variables and edges correspond to an interdependence 	
		between the two connected variables, 
(b) for each vertex $v$, a local transition function that maps tuples of values (belonging to $\XD$) of $v$ and $v$'s 
neighbors to values of $v$, and 
(c)  an update schedule that governs which variables are allowed to update their values in which time steps.
Formal definitions can be found in Sect.~\ref{sec:dynamical-systems}. 

A central goal in the study of dynamical systems is to 
classify them according to how easy it is to predict their behavior. 
In a finite, discrete setting, a certain behavioral pattern is considered predictable if 
it can be decided in polynomial time whether a given system will show the 
pattern \cite{buss-papadimitriou-tsitsiklis-1991}. 
Although the pattern reachability problem is, in general, an intractable problem,
 i.e., at least $\NP$-hard (see, e.g., 
\cite{green-1987,sutner-1995,barrett-hunt-marathe-ravi-rosenkrantz-stearns-2006}), many tractable classes of 
patterns and systems have been identified. However, there is still a serious demand 
for exhaustive characterizations of {\em islands of predictability}.

A fundamental behavioral pattern is the {\em fixed point} (a.k.a., homogeneous state, or equilibrium). 
A value assignment to the variables of a system is a fixed point if the values 
assigned to the variables are left unchanged after the system updates them. 
Note that fixed points are invariant under changes of the update regime. 
In this sense, they can be seen as a particularly robust behavior. A series of recent 
papers has been devoted to the identification of finite 
systems with tractable/intractable fixed-point analyses 
\cite{barrett-hunt-marathe-ravi-rosenkrantz-stearns-tosic-2001,agha-tosic-2005,tosic-2005,tosic-2006,kosub-2008}. 
Precise boundaries are known for which systems finding fixed points can be done in polynomial time. For the fixed-point counting problem 
this is far less so.

\medskip
\noindent
{\em Contributions of the paper.} 
We prove dichotomy theorems on the computational complexity of counting
fixed points in boolean (discrete) dynamical systems, i.e., finite discrete 
dynamical systems over the domain $\{0,1\}$.
For a class $\F$ of boolean functions and a class $\XG$ of graphs, an $(\F,
\XG)$-system is a boolean dynamical system with local transition functions 
lying in $\F$ and a graph lying in $\XG$. 
Following \cite{kosub-2008}, Post classes (a.k.a., clones) and 
forbidden-minor classes are used to classify $(\F,\XG)$-systems. 
In Sect.~\ref{sec:islands} we state the following theorem (Theorem \ref{thm:fpc-t}):  
Let $\F$ be a class of boolean function closed under superposition and let 
$\XG$ be a minor-closed graph class. 
If $\F$ contains all $\min$-functions, all $\max$-functions, or all self-dual 
and monotone functions, and $\XG$ contains all planar graphs, then it is 
$\#\P$-complete to compute the number of the fixed points in an 
$(\F,\XG)$-system; otherwise it is computable in polynomial time. 
Here, the local transition functions are supposed to be given by lookup tables.
In addition, we prove a dichotomy theorem (Theorem 
\ref{thm:fpc-f}) for the case that local transition functions are given by 
formulas (over logical bases). 
Moreover, the corresponding theorem for boolean circuits coincides with the 
theorem for formulas. 
The theorem has a significantly more complicated structure than for 
lookup tables.

\medskip
\noindent
{\em Related work.} 
There is a series of work regarding the complexity of certain computational 
problems for finite discrete dynamical systems (see, e.g., 
\cite{green-1987,sutner-1995,barrett-hunt-marathe-ravi-rosenkrantz-stearns-tosic-2001,barrett-hunt-marathe-ravi-rosenkrantz-stearns-2003,barrett-hunt-marathe-ravi-rosenkrantz-stearns-2003a,agha-tosic-2005,tosic-2005,barrett-hunt-marathe-ravi-rosenkrantz-stearns-2006} and the
references therein). 
The problem of counting fixed points of boolean dynamical systems 
has been studied in \cite{agha-tosic-2005,tosic-2005,tosic-2006}. To summarize:
counting the number of fixed points is in general $\#\P$-complete. So is counting 
the number of fixed points for boolean dynamical systems
with monotone local transition functions over planar bipartite graphs or
over uniformly sparse graphs. We note that all system classes
considered here are based on formula or circuit representations. 
That is, if they fit into our scheme at all, then the intractability results fall 
into the scope of Theorem \ref{thm:fpc-f} (and are covered there).
Detailed studies of computational problems related to fixed-point existence 
have been reported in \cite{barrett-hunt-marathe-ravi-rosenkrantz-stearns-tosic-2001,kosub-2008}. 
In \cite{kosub-2008}, a complete
classification of the fixed-point existence problem with respect to the analysis 
framework we use in this paper was shown.


\section{The Dynamical Systems Framework}
\label{sec:dynamical-systems}

In this section we present a formal framework for dynamical systems. 
A fairly general approach is motivated by the theoretical 
study of simulations. The following is based on  
\cite{barrett-reidys-1999,barrett-mortveit-reidys-2000,barrett-mortveit-reidys-2001,kosub-2008}.

The underlying network structure of a dynamical 
system is given by an undirected graph $G=(V,E)$ without multi-edges and 
loops. We suppose that the set $V$ of vertices is ordered. So, 
without loss of generality, we assume $V=\{1, 2, \dots,n\}$. For any 
vertex set $U\subseteq V$, let $N_G(U)$ denote the neighbors of $U$ in $G$, 
i.e., 
\[N_G(U) =_{\rm def} \{~j~|~j\notin U\textrm{ and there is an $i\in U$ 
such that $\{i,j\}\in E$}~\}.\]
If $U=\{i\}$ for some vertex $i$, then we use $N_G(i)$ as a shorthand for 
$N_G(\{i\})$. The degree $d_i$ of a vertex $i$ is the number of its 
neighbors, i.e., $d_i=_{\rm def} \|N_G(i)\|.$ 

A {\em dynamical system $S$ over a domain $\XD$} is a pair $(G, F)$ where 
$G=(V,E)$ is an undirected graph (the {\em network}) and $F=\{f_i~|~i\in V\}$ 
is a set of {\em local transition functions} $f_i:\XD^{d_i+1}\to \XD$.
The intuition of the definition is that each vertex $i$ corresponds 
to an active element (entity, agent, actor etc.) which is always in some 
state $x_i$ and which is capable to change its state, if necessary. The 
domain of $S$ formalizes the set of possible states of all vertices of the 
network, i.e., for all $i\in V$, it always holds that $x_i\in\XD$. A vector 
$\vec{x}=(x_i)_{i\in V}$ such that $x_i\in \XD$ for all $i\in V$ is called 
a {\em configuration of $S$}. The local transition function $f_i$ for some 
vertex $i$ describes how $i$ changes its state depending on the states of 
its neighbors $N_G(i)$ in the network and its own state. 

We are particularly interested in dynamical systems operating on a discrete time-scale. 
A {\em discrete dynamical system $\XS=(S,\alpha)$} consists of a dynamical system $S$ 
and a mapping $\alpha:\{1,\dots,T\}\to \pot(V)$, where $V$ is a set of vertices of the 
network of $S$ and $T\in\Nat$. The mapping $\alpha$ is called the {\em 
update schedule} and specifies which state updates are realized at certain 
time-steps: for $t\in\{1,\dots,T\}$, $\alpha(t)$ specifies those vertices that 
simultaneously update their states in step $t$. 

A discrete dynamical system $\XS=(S,\alpha)$ over domain 
$\XD$ induces a global map $\XF_\XS:\XD^n \to\XD^n$ where $n$ is the number 
of vertices of $S$. For each vertex $i\in V$, define an {\em activity function} 
$\varphi_i$ for a set $U\subseteq V$ and $\vec{x}=(x_1,\dots,x_n)\in\XD^n$ by
\[
\varphi_i[U](\vec{x})=_{\rm def}\left\{\begin{array}{ll}
f_i(x_{i_1},\dots,x_{i_{d_i+1}})  &\textrm{ if $i\in U$}\\
x_i & \textrm{ if } i\notin U
\end{array}\right.\]
where $\{i_1,i_2,\dots,i_{d_i+1}\}=\{i\}\cup N_G(i)$. For a set $U\subseteq V$, 
define the {\em global transition function} $\XF_S[U]:\XD^n\to \XD^n$ for all 
$\vec{x}\in\XD^n$ by \[\XF_S[U](\vec{x})=_{\rm def} (\varphi_1[U](\vec{x}),
\dots,\varphi_n[U](\vec{x})).\] Note that the global transition function does 
not refer to the update schedule, i.e., it only depends on the dynamical system 
$S$ and not on $\XS$. The function $F_\XS:\XD^n\to\XD^n$ computed by the discrete 
dynamical system $\XS$, the {\em global map} of $\XS$, is defined by
\[\XF_\XS=_{\rm def} \prod_{k=1}^T \XF_S[\alpha(k)].\]

The central notion for our study of dynamical systems is the 
concept of a fixed point, i.e., a configuration which does not change under any 
global behavior of the system.
Let $S=(G,\{f_i~|~i\in V\})$ be a dynamical system over domain $\XD$. 
A configuration $\vec{x}\in \XD^n$ is said to be a {\em local fixed 
point of $S$ for $U\subseteq V$} if and only if $\XF_S[U](\vec{x})=
\vec{x}$.
A configuration $\vec{x}\in \XD^n$ is said to be a {\em fixed point 
of $S$} if and only if $\vec{x}$ is a local fixed point of $S$ for $V$.
Note that a fixed point does not depend on a concrete update schedule:
a configuration $\vec{x}\in\XD^n$ is a fixed point of $S$ if and only if 
for all update schedules $\alpha:\{1,\dots,T\}\to \pot(V)$, it holds that 
$\XF_{(S,\alpha)}(\vec{x})=\vec{x}.$


\section{The Analysis Framework}

In this section we specify our analysis framework for $(\F,\XG)$-systems.
Following \cite{kosub-2008}, local transition functions are classified 
by Post classes, i.e., superpositionally closed classes of boolean functions,
and graphs are classified using the theory of graph minors as a tool. 
In the following we gather relevant notation.

\subsection{Transition Classes}

We adopt notation from \cite{boehler-creignou-reith-vollmer-2003}.
An $n$-ary boolean function $f$ is a mapping $f:\{0,1\}^n\to\{0,1\}$. Let $\BF$
denote the class of all boolean functions. There are two $0$-ary boolean functions: 
$c_0=_{\rm def} 0$ and $c_1=_{\rm def} 1$ (which are denoted in formulas by the symbols 
$0$ and $1$). There are two $1$-ary boolean functions: $\id(x)=_{\rm def} x$ and 
$\nt(x)=_{\rm def} 1-x$ (which are denoted in formulas by $x$ for $\id(x)$ and $\overline{x}$ 
for $\nt(x)$). 

We say that a class $\F$ is {\em Post} if and only if $\F$ contains the function $\id$ and
$\F$ is closed under the introduction of fictive variables, permutations of variables, identification
of variables, and substitution (see, e.g., \cite{boehler-creignou-reith-vollmer-2003} for definitions).
It is a famous theorem by Post \cite{post-1941} that the family of all Post 
classes is a countable lattice with respect to set inclusion.
In particular, each Post class
is the intersection of a finite set of meet-irreducible classes, which are the following:
\begin{itemize}
\item {\em The classes $\BFR_0$ and $\BFR_1$.} For $b\in\{0,1\}$, a boolean 
		function $f$ is said to be {\em $b$-reproducing} if and only if $f(b,
		\dots,b)=b$. Let $\BFR_b$ denote the class of all $b$-reproducing 
		functions. 
\item {\em The class $\BFM$.} For binary $n$-tuples $\vec{a}=(a_1,\dots,a_n)$ 
		and $\vec{b}=(b_1,\dots,b_n)$, we say that $(a_1,\dots,a_n)\le (b_1,
		\dots,b_n)$ if and only if for all $i\in \{1,\dots,n\}$,  it holds that 
		$a_i\le b_i$. An $n$-ary boolean function $f$ is said to be {\em monotone} 
		if and only if for all $\vec{x},\vec{y}\in\{0,1\}^n$, $\vec{x}\le\vec{y}$ 
		implies $f(\vec{x})\le f(\vec{y})$. Let $\BFM$ denote the class of all 
		monotone boolean functions.
\item {\em The class $\BFD$.} An $n$-ary boolean function $f$ is said to be {\em 
		self-dual} if and only if for all $(x_1,\dots,x_n)\in\{0,1\}^n$, it holds 
		that $f(x_1,\dots,x_n)=\nt(f(\nt(x_1),\dots,\allowbreak\nt(x_n)))$. Let 
		$\BFD$ denote the class of all self-dual functions.
\item {\em The class $\BFL$.} A boolean function $f$ is linear if and only if there 
		exists constants $a_1,\dots,a_n\in\{0,1\}$ such that $f(x_1,\dots,x_n)=a_0
		\oplus a_1x_1\oplus \cdots \oplus a_nx_n$. Note that $\oplus$ is understood 
		as addition modulo $2$ and $xy$ is understood as multiplication modulo $2$. 
		Let $\BFL$ denote the class of all linear functions. The logical basis of $\BFL$
		is $\{\oplus,0,1\}$.
\item {\em The classes $\BFS_b$ and $\BFS_b^k$.} For $b\in\{0,1\}$, a tuple set 
		$T\subseteq\{0,1\}^n$ is said to be $b$-separating if and only if there is an
		$i\in\{1,\dots,n\}$ such that for $(t_1,\dots,t_n)\in T$ holds $t_i=b$. A 
		boolean function $f$ is {\em $b$-separating} if and only if $f^{-1}(b)$ is 
		$b$-separating. A function $f$ is called {\em $b$-separating of level 
		$k$} if and only if every $T\subseteq f^{-1}(b)$ such that $\|T\|=k$ 
		is $b$-separating. Let $\BFS_b$ denote the class of $b$-separating 
		functions and let $\BFS_b^k$ denote the class of all functions which 
		are $b$-separating of level $k$.
\item {\em The classes $\BFE$ and $\BFV$.} We denote by $\BFE$ the class of all
		AND functions, i.e., the class of all functions $f$, the arity of which 
		is $n$, such that for some set $J\subseteq\{1,\dots,n\}$, the equality 
		$f(x_1,\dots,x_n)=\min_{i\in J} x_i$ is satisfied for all $x_1,\dots,x_n\in
		\{0,1\}$. The logical basis over $\BFE$ is $\{\land,0,1\}$. Dually, we 
		denote by $\BFV$ the class of all OR functions, i.e.,
		the class of all functions $f$, the arity of which is $n$, such that for
		some set $J\subseteq\{1,\dots,n\}$, the equality $f(x_1,\dots,x_n)=\max_{i
		\in J} x_i$ is satisfied for all $x_1,\dots,x_n\in\{0,1\}$. The logical basis
		of $\BFV$ is $\{\vee,0,1\}$.
\item {\em The class $\BFN$.} An $n$-ary boolean function $f$ is a projection 
		if and only if there is an $i\in\{1,\dots,n\}$ such that for all $x_1,
		\dots,x_n\in\{0,1\}$, it holds that $f(x_1,\dots,x_n)=x_i$. A boolean 
		function $f$ is the negation of a projection if and only if there is an 
		$i\in\{1,\dots,n\}$ such that for all $x_1,\dots,x_n\in\{0,1\}$, it holds 
		that $f(x_1,\dots,x_n)=\nt(x_i)$. A boolean function $f$ is constant if
		and only if there exists a $b\in\{0,1\}$ such that for all $x_1,\dots,
		x_n\in\{0,1\}$, it holds that $f(x_1,\dots,x_n)=b$. Let $\BFN$ denote 
		the class of boolean functions which are projections, negations of 
		projections, or constant functions.
\end{itemize}
Note that the classes possess the following inclusion structure (see, e.g., \cite{boehler-creignou-reith-vollmer-2003}): 
\begin{itemize}
\item $\BFS_0\subset\cdots\subset \BFS_0^k \subset \BFS_0^{k-1}\subset 
		\cdots\subset \BFS_0^2\subset\BFR_1$ 
\item $\BFS_1\subset\cdots\subset \BFS_1^k \subset \BFS_1^{k-1}\subset
		\cdots\subset\BFS_ 1^2\subset \BFR_0$
\item $\BFE\subset\BFM$ and $\BFV\subset\BFM$
\item $\BFN\subset \BFL$
\end{itemize}
No other inclusions hold among these classes. Moreover, all Post classes have a finite logical basis.
Particular relevance for our studies have the following classes:
\[\begin{array}{rll}
\BFD_2	  &=_\df~ \BFD\cap \BFM 			
&\qquad\textrm{with logical basis $\{(x\land y)\vee (x\land z)\vee (y\land z)\}$}\\
\BFS_{00} &=_\df~ \BFS_0 \cap \BFM \cap \BFR_0
&\qquad\textrm{with logical basis $\{x\vee (y\land z)\}$}\\
\BFS_{10} &=_\df~ \BFS_1 \cap \BFM \cap \BFR_1
&\qquad\textrm{with logical basis $\{x\land (y\vee z)\}$}\\
\BFE_2 	  &=_\df~ \BFE\cap \BFS_{10}
&\qquad\textrm{with logical basis $\{\land\}$}\\
\BFV_2    &=_\df~ \BFV\cap \BFS_{00}
&\qquad\textrm{with logical basis $\{\vee\}$}
\end{array}\]

\subsection{Network Classes}

We adopt notation from \cite{diestel-2003}. 
Let $X$ and $Y$ be two undirected
graphs. We say that $X$ is minor of $Y$ if and only if there is a subgraph
 $Y'$ of $Y$ such that $X$ can be obtained by contracting edges of $Y'$.
Let $\preceq$ be the relation on graphs defined by $X\preceq Y$ if and only 
if $X$ is a minor of $Y$. A class $\XG$ of graphs is said to be {\em closed 
under taking minors} if and only if for all graphs $G$ and $G'$, if $G\in\XG$ 
and $G'\preceq G$, then $G'\in \XG$. Let $\XX$ be any set of graphs. 
$\Forb_\preceq(\XX)$ denotes the class of all graphs without a minor in $\XX$ 
(and which is closed under isomorphisms). More specifically, 
$\Forb_\preceq(\XX)=_{\rm def} \{G~|~\textrm{$G\not\succeq X$ for all 
$X\in\XX$}~\}$. The set $\XX$ is called the set of {\em forbidden minors}.
Note that $\Forb_\preceq(\emptyset)$ is the class of all graphs.
As usual, we write $\Forb_\preceq(X_1,\dots,X_n)$ instead of 
$\Forb_\preceq(\{X_1,\dots,X_n\})$. 
Forbidden-minor classes are monotone with respect to $\preceq$, i.e., $X\preceq Y$ 
implies $\Forb_\preceq(X)\subseteq \Forb_\preceq(Y)$.
The celebrated Graph Minor Theorem, due to Robertson 
and Seymour \cite{robertson-seymour-2004}, shows that there are only countably
many network classes closed under taking minors: A class $\XG$ of graphs is 
closed under taking minors if and only if there is a finite set $\XX$ such that 
$\XG=\Forb_\preceq(\XX)$. 

Two graph classes are particularly relevant to our study: planar graphs and graphs
having a vertex cover of size one. Let $K^n$ denote the complete graphs on $n$ 
vertices and let $K_{n,m}$ denote the complete bipartite graph having $n$ vertices 
in one component and $m$ vertices in the other component. The well-known 
Kuratowski-Wagner theorem (see, e.g., \cite{diestel-2003}) states that a graph 
$G$ is planar if and only if $G$ belongs to $\Forb_\preceq(K_{3,3},K^5)$.
Moreover, a graph $X$ is planar if and only if $\Forb_\preceq(X)$ has bounded 
treewidth \cite{robertson-seymour-1986}. As we use the treewidth of a graph only 
in a black-box fashion, we refer to, e.g., \cite{diestel-2003} for a definition. 
A class $\XG$ of graphs is said to have {\em bounded treewidth} if and only if 
there is a $k\in\Nat$ such that all graphs in the class have treewidth at 
most $k$. Let $G=(V,E)$ be a graph. We say that a subset $U\subseteq V$ 
is a {\em vertex cover} of $G$ if and only if for all edges $\{u,v\}\in E$,
it holds that $\{u,v\}\cap U\not=\emptyset$. It is known that the class of graphs 
having a vertex cover of size at most $k$ is closed under taking minors 
\cite{cattell-dinneen-1994}. Moreover, $G$ has a vertex cover of size one if and 
only if $G$ belongs to $\Forb_\preceq(K^3,K^2\oplus K^2)$ \cite{cattell-dinneen-1994},
where for graphs $G$ and $G'$, $G\oplus G'$ denotes the graph obtained by the disjoint 
union of $G$ and $G'$. 
A class of graphs is said to have {\em bounded degree} if and only if there is
a $k\in\Nat$ such that all graphs in the class have a maximum vertex-degree of at
most $k$. It is known that a graph $X$ has a vertex cover of size one if and only if
$\Forb_\preceq(X)$ has bounded degree (cf., e.g., \cite{kosub-2008}).

\section{Islands of Tractability for Fixed Point Counting}
\label{sec:islands}

In this section we are interested in the computational complexity of the following 
counting problem. Let $\F$ be a class of boolean functions and let $\XG$ be a 
class of graphs.
\bigskip

\noindent
\begin{tabular}{p{.15\linewidth}p{.75\linewidth}}
  \textsl{Problem:}  	& $\#\FPE(\F,\XG)$ \\
  \textsl{Input:}    	& An $(\F,\XG)$-system $S$, i.e., a boolean dynamical system $S=(G,\{f_1,\dots,f_n\})$ 
						  such that $G\in\XG$ and for all $i\in\{1,\dots,n\}$, 
						  $f_i\in\F$\\
  \textsl{Output:}   	& The number of fixed points of $S$\\
\end{tabular}

\bigskip
\noindent
The complexity of the problem depends on how transition functions are represented.
We consider the cases of lookup table, formula, and circuit representations. 
The corresponding problems are denoted by $\#\FPE_{\lookup}$, $\#\FPE_\formula$, and
$\#\FPE_\circuit$. 
It is obvious that all problem versions belong to $\#\P$. 
We say that a problem is intractable if it is $\#\P$-hard (with respect to Turing 
reductions, as described in, e.g., \cite{hemaspaandra-ogihara-2002}), and it is 
tractable if it is solvable in polynomial time.

\subsection{The Case of Local Transition Functions Given By Lookup Tables}

We start by identifying tractable counting problems.

\begin{lemma}\label{lem:fpc-l-all-t}
$\#\FPE_\lookup(\BFL,\Forb_\preceq(\emptyset))$ is solvable in polynomial time.
\end{lemma}

\begin{proof}
Notice that for a linear function $f(x_1,\dots,x_n)=a_0\oplus a_1x_2\oplus a_2x_2
\oplus\cdots\oplus a_nx_n$, the proposition $x_i\leftrightarrow [a_0\oplus a_1x_2
\oplus a_2x_2\oplus\cdots\oplus a_nx_n]$ is true if and only if $a_0\oplus a_1x_2\oplus 
a_2x_2\oplus\cdots\oplus a_nx_n\oplus x_i\oplus 1$ is satisfiable. So, each 
dynamical system with linear, boolean local transition functions constitutes a system of 
linear equations over $Z_2$, for which the number of solutions can be computed
in polynomial time using Gaussian elimination (cf.~\cite{creignou-hermann-1996}).
\end{proof}

In \cite{kosub-2008}, it has been shown that the decision version of 
$\#\FPE_\lookup(\BF,\allowbreak\Forb_\preceq(X))$ for planar graphs $X$ can be solved in 
polynomial time. 
This result is obtained by a reduction to a certain type of constraint satisfaction 
problems. 
Actually, the reduction establishes injections between the fixed points of a 
dynamical system and the satisfying assignments of the corresponding constraint 
satisfaction problem. 
Consequently, the numbers of fixed points and the numbers of satisfying 
assignments are equal. 

\begin{lemma}\label{lem:fpc-bf-planar-t}
Let $X$ be a planar graph. Then, $\#\FPE_\lookup(\BF,\Forb_\preceq(X))$ is solvable in polynomial time.
\end{lemma}

\begin{proof} By inspection of \cite{kosub-2008} and noting that counting 
satisfying assignments for constraint satisfaction problems having constraint graphs 
of bounded treewidth can be done in polynomial time (cf.~\cite{flum-grohe-2004}).
\end{proof}

We turn to the intractable fixed-point counting problems. Let $H$ be a 2CNF such that each 
clause consists of exactly one positive and one
negative literal. $H$ is called a Horn-2CNF formula. Moreover, suppose $H$ 
has a planar graph representation, i.e., the graph $\Gamma(H)=(V,E)$ with 
vertex set $V=\{x_1,\dots,x_n,C_1,\dots,C_m\}$, where the $x_i$'s are the variables 
and the $C_i$'s are the clauses of $H$, and edge set 
$E=\{ \{x_i,C_j\}~|~\textrm{ $x_i$ is a variable in $C_j$}\}$ is planar. 
Then, $H$ is called a planar Horn 2-CNF formula.
$\#\PHSAT$ is the problem of counting all satisfying assignments of a given planar 
Horn-2CNF formula. 

\begin{proposition}
$\#\PHSAT$ is $\#\P$-complete even if each variable is allowed to occur in four 
clauses only.
\end{proposition}

\begin{proof} In \cite{vadhan-2002}, it has been shown that the following 
\sloppy problem is $\#\P$-complete: $\#4\Delta\textsc{-Planar}\allowbreak
\textsc{Bipartite Independent Set}$, i.e., compute, on a given bipartite graph 
$G=(V,E)$ with maximum vertex-degree at most four, the number of independent 
sets $U\subseteq V$. Let $G=(V,E)$ be a bipartite graph, $V=V_1\cup V_2$ and 
$E\subseteq V_1\times V_2$. Define $H$ to be the 2CNF given by clauses
$(x_u\vee \overline{x_v})$ for all $u\in V_1$, $v\in V_2$ such that $\{u,v\}\in E$.
Clearly, $H$ is a Horn-2CNF formula. Moreover, if $G$ is planar and the maximum
degree is at most four, the graph representation of $H$ is planar and each variable
occurs at most four times in $H$. Finally, it is easily seen that there is a bijection
between the independent sets of $G$ and the satisfying assignments for $H$ (cf., e.g., 
\cite{ball-provan-1983,linial-1986}). Hence, $\#4\Delta\textsc{-Planar}\allowbreak
\textsc{Bipartite Independent Set}$ reduces to $\#\PHSAT$ with each variable occuring 
in at most four clauses.
\end{proof}

\begin{lemma}\label{lem:fpc-e2-planar-t}
$\#\FPE_\lookup(\BFE_2,\Forb_\preceq(K_{3,3},K^5))$ is $\#\P$-complete.
\end{lemma}

\begin{proof}
We reduce from $\#\PHSAT$ assuming 
that each variable occurs only four times in the formula. Let $H=C_1\land\dots\land C_m$ be a 
planar Horn-2CNF formula. Define a dynamical system $S=(G,F)$ as follows. 
$G=(V,E)$ is given by $V=_{\rm def} \{1,\dots,n\}$ and $E=_{\rm def} \{~\{i,j\}~|~
\textrm{$(\overline{x_i}\vee x_j)=C_r$ for}\allowbreak\textrm{ some $r\in\{1,\dots,m\}$}~\}$. Since
$H$ has a planar graph representation, $G$ is planar, i.e., $G\in
\Forb_\preceq(K_{3,3},K^5)$. The local transition functions are specified 
in the following way. For a vertex $i_0\in V$ let $\{i_1,\dots,i_r\}$ be the set of
all vertices such that $(\overline{x_{i_j}}\vee x_{i_0})$ is a clause in $H$. Then,
$f_{i_0}$  is the function given by the formula $H_{i_0}=x_{i_0}\land x_{i_1}\land
\dots\land x_{i_r}$. Notice that all local transition functions belong to $\BFE_2$ and
also notice that the maximum degree of a vertex in $G$ is four. Thus,
we can compute the lookup tables in polynomial time depending on the size of $H$.
Moreover, it is easily seen that 
$(x_{i_0}\leftrightarrow \bigwedge_{j=1}^r x_{i_j})
~\equiv~ \bigwedge_{j=1}^r (\overline{x_{i_j}}\vee x_{i_0}).$ Hence,
the number of satisfying assignments of $H$ is equal to the number of fixed-point
configurations of $S_H$. This shows that 
$\#\PHSAT$ reduces to $\#\FPE_\lookup(\BFE_2,\Forb_\preceq(K_{3,3},K^5))$.
\end{proof}

\begin{lemma}\label{lem:fpc-v2-planar-t}
$\#\FPE_\lookup(\BFV_2,\Forb_\preceq(K_{3,3},K^5))$ is $\#\P$-complete.
\end{lemma}

\begin{proof}
Again we reduce from $\#\PHSAT$ assuming 
that each variable occurs only four times in the formula. Let $H=C_1\land\dots\land C_m$ be a 
planar Horn-2CNF formula.
We construct the same network as in the proof of Lemma \ref{lem:fpc-e2-planar-t} 
on a given planar Horn-2CNF formula $H$ having the variables $x_1,\dots,x_n$. However, 
the local transition functions are specified as follows. For a vertex $i_0\in V$,
let $\{i_1,\dots,i_r\}$ be the set of all vertices such that $(\overline{x_{i_0}}
\vee x_{i_j})$ is a clause in $H$. Then, $f_{i_0}$  is the function given by 
the formula $H_{i_0}=x_{i_0}\vee x_{i_1}\vee \dots\vee x_{i_r}$ which clearly 
belongs to $\BFV_2$. It remains to verify the number of satisfying assignments
of $H$ equals the number of fixed-point configuration of $S_H$. This follows from
$(x_{i_0}\leftrightarrow \bigvee_{j=1}^r x_{i_j})
~\equiv~ \bigwedge_{j=1}^r (\overline{x_{i_0}}\vee x_{i_j}).$ Hence, 
$\#\PHSAT$ reduces to $\#\FPE_\lookup(\BFV_2,\Forb_\preceq(K_{3,3},K^5))$.
\end{proof}

We now turn our attention to proving that $\#\FPE_\lookup(\BFD_2,
\Forb_\preceq(K_{3,3},\allowbreak K^5))$ is intractable. 
Our proof is based on a reduction from the following problem, shown in~\cite{vadhan-2002} 
to be intractable: $\#4\Delta\textsc{-Planar Bipartite Vertex}\allowbreak \textsc{Cover}$, i.e.,  
compute, on a given planar bipartite graph $G$ of maximum vertex-degree at most four,
the number of vertex covers in $G$.
The reduction uses the following gadget, which increases the number of fixed points by a very large 
factor whenever two particular variables are unequal. 

\begin{definition}
For $h\in\Nat$, an \emph{$h$-amplifier} is a dynamical system $(G,F)$, where 
$G = (V,E)$, such that
\begin{eqnarray*} 
V &=_\df& \{a_0,\ldots,a_h\} ~ \cup ~ \{b_0,\ldots,b_h\} ~ \cup ~ \{c_0,\ldots, c_h\}\\
E &=_\df& \{\{u_r,j_r\}~|~u \in \{a,c\} \text{ and } r \in \{0,\ldots,h\}\} \\
  & & \cup ~ \{\{u_r, u_{r-1}\}~|~u \in \{a,c\} \text{ and } r \in \{1,\ldots,h\}\} \\
F &=_\df& \{f_{u_0}~|~u \in \{a,c\} \text{ and } f_{u_0} =_{\df} x_{u_0} \} \\
  & & \cup ~ \{f_{u_r}~|~u \in \{a,c\}, r \in \{1,\ldots, h\}, \text{ and } f_{u_r} =_{\df} x_{u_{r-1}}\} \\
  & & \cup ~ \{f_{b_r}~|~r \in \{0,\ldots,h\} \text{ and } f_{b_r} =_{\df} (x_{b_r} \vee x_{a_r}) \wedge (x_{b_r} \vee x_{c_r}) \wedge (x_{a_r} \vee x_{c_r}))\}.
\end{eqnarray*}\end{definition}


\begin{proposition}\label{prop:amp}
For each $h$-amplifier $A_h$, there is exactly one fixed point whenever $x_{a_0} = x_{c_0}$ and 
$2^{h+1}$ fixed points otherwise.
\end{proposition}
\begin{proof} 
Note that by the definitions of the update functions there is a fixed point of $A_h$ 
if and only if $x_{a_0} = x_{a_1} = \cdots = x_{a_h}$, $x_{c_0} = x_{c_1} = \cdots = x_{c_h}$, and
$(x_{a_0} \neq x_{c_0}) \vee (x_{b_0} = x_{b_1} = \cdots = x_{b_h})$.\end{proof}

We now prove the intractability result.
\begin{lemma}\label{lem:d2}
$\#\FPE_\lookup(\BFD_2,\Forb_\preceq(K_{3,3},K^5))$ is intractable.
\end{lemma}
\begin{proof}
Let $H = (U,D)$ be a planar, bipartite graph of degree at most four, where $m = \|U\|$ and $n = \|D\|$.
Define $S = (G,F)$ to be the dynamical system where $G = (V,E)$, $V =_\df U \cup D$,
$E =_\df \{\{i, \{i,j\}\}~|~\{i,j\} \in D\}$, and 
\begin{eqnarray*}
F &=_\df& \{f_i~|~i \in U \text{ and } f_i=_{\df} x_i\} \\
  &  & \cup ~ \{f_{\{i,j\}}~|~\{i,j\} \in D \text{ and } f_{\{i,j\}} =_\df (x_{\{i,j\}} \vee x_i) \wedge (x_{\{i,j\}} \vee x_j) \wedge (x_i \vee x_j)\}.
\end{eqnarray*}
Clearly the graph $G$ is planar, as it has $H$ as a topological minor and the update functions in $F$ are in $\BFD_2$. 
How do the fixed points in $S$ correspond to the vertex covers of $H$? Note that, for each $\{i,j\} \in D$, 
$f_{\{i,j\}}(x_i,x_{\{i,j\}},x_j) = x_{\{i,j\}} \Longleftrightarrow x_{\{i,j\}} = x_i \vee x_{\{i,j\}} = x_j$.
We want to regard edge $\{i,j\}$ as being covered in $H$ whenever this is so. 
Indeed, the number of fixed points in $S$ such that all the variables in $\{x_{\{i,j\}}~|~\{i,j\} \in D\}$ are 
equal is twice the number of vertex covers that $H$ has. 
Of course, $S$ may have additional fixed points, i.e., fixed points where the values of the variables in 
$\{x_{\{i,j\}}~|~\{i,j\} \in D\}$ are not equal. 
These ``bad'' fixed points do not correspond to vertex covers in $H$. 
To help ``filter'' the bad fixed points out, we add $h$-amplifiers for sufficiently large $h$ to $S$.

Fix a planar layout of $G$. 
For each $i$, $j$, and $k$ such that the vertices $\{i,j\}$, $j$, and $\{j,k\}$ all lie on the same boundary 
of some face in the layout, add an $(m+1)$-amplifier by identifying $a_0$ with $\{i,j\}$ and $c_0$ with $\{j,k\}$ (or 
\emph{vice-versa}). 
Call the resulting dynamical system $S'$. 
By Proposition~\ref{prop:amp}, $S'$ has exactly one fixed point for each fixed point of $S$ where the edge 
variables $\{x_{\{i,j\}}~|~\{i,j\} \in D\}$ are all equal. 
Note that there are at most $2^{m+1}$ such fixed points. 
For any fixed point in $S$ where the edge variables are unequal, there must exist $i$, $j$, and $k$ such that 
$\{i,j\}$, $j$, and $\{j,k\}$ all lie on the same boundary of some face in the layout and $x_{\{i,j\}} \neq x_{\{j,k\}}$. 
But then, by Proposition~\ref{prop:amp} the number of fixed points in $S'$ that correspond to this fixed point in
$S$ is a multiple of $2^{m+2}$.
So twice the number of vertex covers of $H$ is equal to the number of fixed points in $S'$ modulo $2^{m+2}$. 
Note that the graph in $S'$ is planar (as amplifiers are always planar) and no update function has more than ten 
arguments, so the lookup table can be constructed in polynomial time. 
Note also that each of the update functions in $S'$ is in $\BFD_2$.
Hence, $\#4\Delta\textsc{-Planar Bipartite Vertex Cover}$ reduces to $\#\FPE(\BFD_2,\Forb_\preceq(K_{3,3},K^5))$
in polynomial time.
\end{proof}



Finally, we combine the results to obtain the following conditional 
dichotomy theorem.

\begin{theorem}\label{thm:fpc-t}
Let $\F$ be a Post class of boolean functions and let $\XG$ be a graph class 
closed under taking minors. 
If ($\F\supseteq \BFV_2$ or $\F\supseteq \BFE_2$ or $\F\supseteq\BFD_2$) 
and $\XG\supseteq\Forb_\preceq(K_{3,3},K^5)$, then $\#\FPE_\lookup(\F,\XG)$ is 
intractable, otherwise $\#\FPE_\lookup\allowbreak(\F,\XG)$ is tractable.
\end{theorem}

\begin{proof}
If ($\F\supseteq \BFV_2$ or $\F\supseteq \BFE_2$ or $\F\supseteq\BFD_2$) 
and $\XG\supseteq\Forb_\preceq(K_{3,3},K^5)$, then $\#\FPE_\lookup\allowbreak
(\F,\XG)$ is $\#\P$-complete by Lemma \ref{lem:fpc-e2-planar-t}, Lemma 
\ref{lem:fpc-v2-planar-t}, and by the assumption made for $\BFD_2$. 
Suppose the premise is not satisfied. 
First, assume that $\F\not\subseteq\BFV_2$, $\F\not\subseteq\BFE_2$, and 
$\F\not\subseteq \BFD_2$.
The maximal Post class having this property is $\BFL$.  
By Lemma \ref{lem:fpc-l-all-t}, $\#\FPE_\lookup(\BFL,\Forb_\preceq(\emptyset))$ 
is tractable. 
It remains to consider the case $\XG\not\supseteq\Forb_\preceq(K_{3,3},K^5)$. 
That is $\XG\subseteq\Forb_\preceq(X)$ for some planar graph $X$. 
Lemma \ref{lem:fpc-bf-planar-t} shows that $\#\FPE_\lookup(\BF,\XG)$ is 
solvable in polynomial time.
\end{proof}

\subsection{Succinctly Represented Local Transition Functions}

In this section we prove a dichotomy theorem for the fixed-point counting 
problem when transition are given by formulas or circuits. As usual, the size 
of formula is the number of symbols from the basis used to encode the formula, 
the size of a circuit is the number of gates (from the basis) it consists of 
(including the input gates).\footnote{
Note that, though the fan-in's of the logical bases cannot be bounded by
one constant for all Post classes, for each Post class there is a logical basis 
of bounded fan-in. In particular, those classes which occur in the proofs of this 
section have bases of fan-in at most three.}
Both succinct representations of functions lead to the same 
result.
We only prove special results for the case of formula representations. The 
corresponding results for circuit representations follow easily.

Again we start with gathering the tractable cases.

\begin{lemma}\label{lem:fpc-l-all-f}
$\#\FPE_\formula(\BFL,\Forb_\preceq(\emptyset))$ is solvable in polynomial time.
\end{lemma}

\begin{proof}
Similar to the proof of Lemma \ref{lem:fpc-l-all-t} by noting
that each boolean circuit $C$ over the base $\{\oplus,1,0\}$ can be easily
transformed (in polynomial time in the number of gates of $C$) into the described
system of linear equations over $Z_2$.
\end{proof}

\begin{lemma}\label{lem:fpc-e-planar-f}
Let $X$ be a planar graph. Then, $\#\FPE_\formula(\BFE,\Forb_\preceq(X))$ is
solvable in polynomial time.
\end{lemma}

\begin{proof} Since $X$ is planar, there exists a $k\in\Nat$ such that for all $G\in
\Forb_\preceq(X)$, the treewidth of $G$ is at most $k$. 
Let $S=(G,\{f_1,\dots,f_n\})$ be a dynamical system such that $G=(V,E)\in 
\Forb_\preceq(X)$ and for all $i\in V$, the local transition function $f_i$ is one of 
the constant functions $c_0$ or $c_1$, or is represented by a formula 
$H_i=\bigwedge_{j\in J_i} x_j$, where $J_i\subseteq N_G(i)\cup \{i\}$. 
Without loss of generality, we may assume that there is no $i\in V$ such that 
$f_i\equiv c_1$ or $f_i\equiv c_0$. 
(Otherwise, an obvious procedure exists to eliminate such vertices.) 
We define the directed graph $A(S)$ to consist of $S$'s vertex set $V$ and the 
edge set $E'=_{\rm def} \{~(i,j)~|~i,j\in V, i\in J_j~\}$. 
Note that $A(S)$ is allowed to have loops. 
Observe that for all vertices $i,j\in V$ and all fixed-point configurations 
$\vec{x}$ it holds that if $x_i=0$ then $x_j=0$. An easy consequence is that 
if $C=\{i_1,\dots,i_r\}$ is a strongly connected component of $A(S)$ and $\vec{x}$ is 
a fixed-point configuration, then $x_{i_1}=\cdots=x_{i_r}$. 
Let $\{C_1,\dots,C_\ell\}$ be the set of all strongly connected components of 
$A(S)$. 
Then, the number of fixed-point configurations of $S$ is equal to the number of
satisfying assignments of the constraint satisfaction problem $\CSP(S)=(W,\XD,\XC)$ 
defined by $W=_{\rm def} \{x_1,\dots,x_\ell\}$, $\XD =_{\rm def} \{0,1\}$, and
$\XC =_{\rm def} \{~Ex_ix_j~|~\textrm{ there are $u\in C_i$ and $v\in C_j$ such 
that $(u,v)\in E'$}~\}$ where for all $i,j$ such that $Ex_ix_j\in\XC$, $E_{ij}=_{\rm def} 
\{~(0,0), (1,0), (1,1)~\}$\footnote{A constraint satisfaction problem (CSP) 
consists of triples $(X,\XD,\XC)$, where $X=\{x_1,\dots,x_n\}$ is the set of 
variables, $\XD$ is the domain of the variables, $\XC$ is a set of constraints
$Rx_{i_1},\dots,x_{i_k}$ having associated corresponding relations 
$R_{{i_1},\dots,{i_k}}$. The set $\XC$ of constraints is listed by pairs 
$\langle Rx_{i_1},\dots,x_{i_k},\allowbreak R_{{i_1},\dots,{i_k}} \rangle$.
A solution for $(X,\XD,\XC)$ is an assignment $I:X\to\XD$ such that
$(I(x_{i_1}),\dots,I(x_{i_k}))\in R_{{i_1},\dots,{i_k}}$ for all constraints
$Rx_{i_1},\dots,x_{i_k}\in \XC$. The (primal) constraint graph for $(X,\XD,\XC)$ 
consists of the vertex set $X$ and the edge set $\{ \{x_i,x_j\}~|~\textrm{$x_i$ 
and $x_j$ occur in the same constraint of $\XC$}\}$.}.
Note that the constraint graph of $\CSP(S)$ (up to being oriented) is a minor 
of the network of $S$. It follows that the constraint graph
has treewidth at most $k$. Hence, using the algorithms in \cite{flum-grohe-2004},
the number of fixed-point configurations can be computed in polynomial time. 
Consequently, $\#\FPE_\formula(\BFE,\Forb_\preceq(X))$ can be solved in polynomial time.
\end{proof}

\begin{lemma}\label{lem:fpc-v-planar-f}
Let $X$ be a planar graph. Then, $\#\FPE_\formula(\BFV,\Forb_\preceq(X))$ is
solvable in polynomial time.
\end{lemma}

\begin{proof}
The case of $\BFV$ is dual to the case of $\BFE$. Indeed, suppose we have
a dynamical system $S=(G,\{f_1,\dots,f_n\})$ such that $G=(V,E)\in\Forb_\preceq(X)$ 
and for all $i\in V$, $f_i$ is constant or represented by a formula $H_i=
\bigvee_{j\in J_i} x_j$ where $J_i\subseteq N_G(i)\cup \{i\}$. Replace
each $\vee$ by $\land$, $0$ by $1$, and $1$ by $0$. Obviously, this gives
a dynamical system having the same number of fixed-point configurations as $S$.
Thus, $\#\FPE_\formula(\BFV,\Forb_\preceq(X))$ reduces to 
$\#\FPE_\formula(\BFE,\Forb_\preceq(X))$. 
Hence, by Lemma \ref{lem:fpc-e-planar-f},
$\#\FPE_\formula(\BFV,\allowbreak\Forb_\preceq(X))$ can be solved in polynomial time.
\end{proof}

\begin{lemma}\label{lem:fpc-bf-bdegree-f}
\sloppy Let $X$ be a graph with a vertex cover of size one. Then, $\#\FPE_\formula\allowbreak
(\BF,\Forb_\preceq(X))$ is solvable in polynomial time.
\end{lemma}

\begin{proof}
Let $X$ have a vertex cover of size one, i.e., $\Forb_\preceq(X)$ has bounded 
degree. 
\sloppy So, it is easily seen that for all classes $\F$ of boolean functions, 
$\#\FPE_\formula(\F,\Forb_\preceq(X))$ reduces to $\#\FPE_\lookup(\F,\Forb_\preceq(X))$. 
As $X$ is also a planar graph (note that $K^3\preceq K_{3,3}$ and
$K^3\preceq K^5$ as well as $K^2\oplus K^2\preceq K_{3,3}$ and $K^2\oplus K^2\preceq K^5$), $\#\FPE_\formula\allowbreak(\BF,\Forb_\preceq(X))$ is solvable in polynomial time
using Lemma \ref{lem:fpc-bf-planar-t}.
\end{proof}

We turn to the $\#\P$-complete cases.

\begin{lemma}\label{lem:fpc-e2-planar-f}
$\#\FPE_\formula(\BFE_2,\Forb_\preceq(K_{3,3},K^5))$ is $\#\P$-complete.
\end{lemma}

\begin{proof}
An inspection of the proof of Lemma \ref{lem:fpc-e2-planar-t} shows that the 
local transition functions specified there are in fact, represented by formulas. 
Thus, the proposition follows from the proof of Lemma \ref{lem:fpc-e2-planar-t}.
\end{proof}

\begin{lemma}\label{lem:fpc-v2-planar-f}
$\#\FPE_\formula(\BFV_2,\Forb_\preceq(K_{3,3},K^5))$ is $\#\P$-complete.
\end{lemma}

\begin{proof}
Similar to Lemma \ref{lem:fpc-e2-planar-f} by inspecting the proof of Lemma \ref{lem:fpc-v2-planar-t}.
\end{proof}

Let $H$ be a 2CNF formula such that each clause consists of positive literals only.
$H$ is called a positive 2CNF. It is well known that the counting problem $\#\POSSAT$
, i.e., counting the satisfying assignments of positive 2CNF, is $\#\P$-complete 
\cite{valiant-1979}.

\begin{lemma}\label{lem:fpc-s10-bdegree-f}
$\#\FPE_\formula(\BFS_{10},\Forb_\preceq(K^3,K^2\oplus K^2))$ is $\#\P$-complete.
\end{lemma}

\begin{proof}
We reduce from $\#\POSSAT$.
Let $H=C_1\land\dots\land C_m$ be a positive 2CNF formula having variables 
$x_1,\dots, x_n$. 
Let $\#_+(H)$ denote the number of satisfying assignments of $H$. 
Let $S_{10}(x,y,z)=_{\rm def} (x\land (y\vee z))$ denote the only element in the 
logical basis of $S_{10}$. 
Define $S_H$ to be the dynamical system consisting of the network $G=(V,E)$, 
where $V=_{\rm def} 
\{1,\dots,n,n+1\}$ and $E=_{\rm def} \{~\{i,n+1\}~|~i\in\{1,\dots,n\}~\}$, and 
the local transition functions are specified as follows. For $i\in\{1,\dots,n\}$
set $P_i(x_i,x_{n+1})=_{\rm def} S_{10}(x_i,x_i,x_i)$ and let $f_i$ be represented 
by $P_i$. For $i=n+1$, we first define auxiliary formulas $A_j$ for $j\in\{1,\dots,m\}$
by $A_1(x_1,\dots,x_{n+1})=_{\rm def} S_{10}(x_{n+1},x_{11},x_{12})$ and 
for $k>1$ by $A_k(x_1,\dots,x_{n+1})=_{\rm def} S_{10}(A_{k-1}(x_1,\dots,x_{n+1}),
x_{k1},x_{k2})$ where $C_k=(x_{k1}\vee x_{k2})$. Finally, set $P_{n+1}(x_1,\dots,
x_{n+1})=_{\rm def} A_m(x_1,\dots,x_{n+1})$ and let $f_{n+1}$ be represented by 
$P_{n+1}$. 
Certainly, $S_H$ is an $(\BFS_{10},\Forb_\preceq(K^3,K^2 \oplus K^2))$-system 
computable in time polynomial in the size of $H$. 
Moreover, note that $P_{n+1}(x_1,\dots,x_{n+1})\equiv x_{n+1}\land \bigwedge_{j=1}^m C_j$.
It follows that the number of fixed-point configurations of $S_H$ is $\#_+(H)+2^n$. 
Hence, $\#\POSSAT$ reduces to $\#\FPE_\formula(\BFS_{10},\Forb_\preceq(K^3,K^2\oplus K^2))$.
\end{proof}

\begin{lemma}\label{lem:fpc-s00-bdegree-f}
$\#\FPE_\formula(\BFS_{00},\Forb_\preceq(K^3,K^2\oplus K^2))$ is $\#\P$-complete.
\end{lemma}

\begin{proof}
Again we reduce from $\#\POSSAT$.
Let $H=C_1\land\dots\land C_m$ be a positive 2CNF formula having variables 
$x_1,\dots, x_n$. 
Let $\#_+(H)$ denote the number of satisfying assignments of $H$. 
Let $S_{00}(x,y,z)=_{\rm def} (x\vee (y\land z))$ denote the only element in the 
logical basis of $S_{00}$. We define $S_H$ to be the dynamical system 
consisting of the network $G=(V,E)$, where $V=_{\rm def} \{0,1,\dots,n,n+1\}$ 
and $E=_{\rm def}\{~\{i,n+1\}~|~i\in\{0,\dots,n\}~\}$, and the set of local 
transition functions specified as follows: for $i\in\{0,\dots,n\}$, set 
$P_i(x_i,x_{n+1})=_{\rm def} S_{00}(x_i,x_i,x_i)$ and let $f_i$ be the function
represented by $P_i$. For $i=n+1$, i.e., the center of the star $G$, we first
introduce auxiliary formulas $A_{j_1,\dots,j_k}(x_0,x_1,\dots,x_n)$ for $k\in\Nat_+$
and $j_1<\cdots<j_k$ inductively defined by $A_i(x_0,x_1,\dots,x_n)=_{\rm def} 
S_{00}(x_{i_1},x_{i_2},x_{i_2})$, such that $C_i=(x_{i_1}\vee x_{i_2})$, and
\begin{eqnarray*}
\lefteqn{A_{j_1,\dots,j_k}(x_0,\dots,x_n)=_{\rm def}} \\
& & ~~~~~~~~~~~~~~~~ S_{00}(x_0,A_{j_1,\dots,j_{\lfloor 
k/2\rfloor}}(x_0,\dots,x_n),A_{j_{\lfloor k/2\rfloor +1},\dots,j_k}(x_0,\dots,
x_n)).
\end{eqnarray*}
\sloppy We finally define $P_{n+1}(x_0,\dots,x_{n+1})=_{\rm def} S_{00}(x_0,
A_{1,\dots,m}(x_0,\dots,x_n),x_{n+1})$. Clearly, $S_H$ is an $(S_{00},
\Forb_\preceq(K^3,K^2\oplus K^2))$-system computable in time polynomial 
in the size of $H$. Moreover, as is easily seen by induction it holds that
\[\textstyle 
A_{j_1,\dots,j_k}(c_0,x_1,\dots,x_n) \equiv \bigwedge_{\ell=1}^k C_\ell
~~\textrm{ and }~~
A_{j_1,\dots,j_k}(c_1,x_1,\dots,x_n) \equiv c_1.
\]
This leads to the following numbers of fixed-point configurations of $S_H$:
\[\begin{array}{l}
\textrm{- there are $2^n$ fixed-point configurations $\vec{x}$ such that $x_0=0$ and
	  $x_{n+1}=0$,}\\[.5ex]
\textrm{- there are $\#_+(H)$ fixed-point configurations $\vec{x}$ such that $x_0=0$ 
	  and $x_{n+1}=1$,}\\[.5ex]
\textrm{- there is no fixed-point configuration $\vec{x}$ such that $x_0=1$ and
	  $x_{n+1}=0$, and}\\[.5ex]
\textrm{- there are $2^n$ fixed-point configurations $\vec{x}$ such that $x_0=1$ and
	  $x_{n+1}=1$.}
\end{array}\]
Hence, the number of fixed-point configurations of $S_H$ is just $\#_+(H)+2^{n+1}$.
Consequently, $\#\POSSAT$ reduces to $\#\FPE_\formula(\BFS_{00},\Forb_\preceq(K^3,K^2\oplus K^2))$.
\end{proof}

\begin{lemma}\label{lem:fpc-d2-bdegree-f}
$\#\FPE_\formula(\BFD_2,\Forb_\preceq(K^3,K^2\oplus K^2))$ is $\#\P$-complete.
\end{lemma}

\begin{proof} 
We reduce from $\#\POSSAT$. 
We construct the same network as for the case $\BFS_{00}$ in the proof 
Lemma \ref{lem:fpc-s00-bdegree-f} on a given positive 2CNF formula 
$H=C_1\land\dots\land C_m$ having variables $x_1,\dots,x_n$. 
Let $\#_+(H)$ denote the number of satisfying assignments of $H$. 
The local transition functions are specified
as follows. Let $D_2(x,y,z)=_{\rm def} (x\land y)\vee (x\land z)\vee (y\land z)$
denote the only element in the logical basis of $D_2$. For $i\in\{0,\dots,n\}$ set 
$P_i(x_i,x_{n+1})=_{\rm def} D_2(x_i,x_i,x_i)$ and let $f_i$ be represented by 
$P_i$. For $i=n+1$, we again introduce auxiliary 
$A_{j_1,\dots,j_k}(x_0,x_1,\dots,x_n,x_{n+1})$ for $k\in\Nat_+$
and $j_1<\cdots<j_k$ inductively defined by $A_i(x_0,\dots,x_{n+1})=_{\rm def} 
D_2(x_{i_1},x_{i_2},x_{n+1})$, such that $C_i=(x_{i_1}\vee x_{i_2})$, and
\begin{eqnarray*}
\lefteqn{A_{j_1,\dots,j_k}(x_0,\dots,x_{n+1}) =_{\rm def}} \\
& & ~~~~~~~~~~~~~~ D_2(A_{j_1,\dots,j_{\lfloor k/2\rfloor}}(x_0,\dots,x_{n+1}),A_{j_{\lfloor k/2\rfloor +1},\dots,j_k}(x_0,\dots,
x_{n+1}), x_0).
\end{eqnarray*}
We finally define $P_{n+1}=_{\rm def} A_{1,\dots,m}$. Evidently,
$S_H$ is a $(\BFD_2,\Forb_\preceq(K^3,K^2\oplus K^2))$-system and can be 
computed in time polynomial in the size of $H$. Moreover, by induction over the formula
structure of $P_{n+1}$ we easily obtain the following equivalences:
\begin{eqnarray*}
P_{n+1}(0,x_1,\dots,x_n,0) &\equiv& \textstyle \bigwedge_{i=1}^n x_i\\
P_{n+1}(0,x_1,\dots,x_n,1) &\equiv& \textstyle \bigwedge_{i=1}^m (x_{i1}\vee x_{i2})\\
P_{n+1}(1,x_1,\dots,x_n,0) &\equiv& \textstyle \bigvee_{i=1}^m (x_{i1}\land x_{i2})\\
P_{n+1}(1,x_1,\dots,x_n,1) &\equiv& \textstyle \bigvee_{i=1}^n x_i
\end{eqnarray*}
Thus, the number of fixed-point configurations of $S_H$ is exactly $2\#_+(H) + 2^{n+1}-2$.
Hence, $\#\POSSAT$ reduces to $\#\FPE_\formula(\BFD_2,\Forb_\preceq(K^3,K^2\oplus K^2))$.
\end{proof}

Finally, we combine all results to obtain the following dichotomy theorem.

\begin{theorem}\label{thm:fpc-f}
Let $\F$ be a Post class of boolean functions and let $\XG$ 
be a graph class closed under taking minors. Then, $\#\FPE_\formula(\F,\XG))$ is 
intractable if one of the following conditions is satisfied.
\begin{enumerate}
\item $\bigl(\F\supseteq \BFS_{00}$ or $\F\supseteq \BFS_{10}$ or $\F\supseteq 
	\BFD_2\bigr)$ and $\XG\supseteq\Forb_\preceq(K^3,K^2\oplus K^2)$.
\item $\bigl(\F\supseteq \BFV_2$ or $\F\supseteq \BFE_2\bigr)$ and $\XG\supseteq
	\Forb_\preceq(K_{3,3}, K^5)$.
\end{enumerate}
Otherwise, $\#\FPE_\formula(\F,\XG)$ is tractable. Moreover, the same classification 
is true for $\#\FPE_\circuit(\F,\XG)$.
\end{theorem}

\begin{proof}
If for $\F$ and $\XG$ the first conditions is satisfied, then the intractability follows from
Lemmas \ref{lem:fpc-s10-bdegree-f}, \ref{lem:fpc-s00-bdegree-f}, and \ref{lem:fpc-d2-bdegree-f}.
If $\XG\not\supseteq\Forb_\preceq(K^3,K^2\oplus K^2)$, then,
as argued in \cite{kosub-2008}, there is a graph $X$ having a 
vertex cover of size one such that $\XG\in\Forb_\preceq(X)$. 
Lemma \ref{lem:fpc-bf-bdegree-f} shows that $\#\FPE_\formula(\BF,\Forb_\preceq(\XG))$ is solvable 
in polynomial time. 
Assume that $\F\not\supseteq\BFS_{00}$, $\F\not\supseteq\BFS_{10}$, and $\F\not\supseteq \BFD_2$. 
The maximal Post classes satisfying this are $\BFV$, $\BFE$, and $\BFL$. 
Thus, we only consider subclasses of these three classes. If $\F$ and $\XG$ satisfy 
the second condition, then the Lemmas \ref{lem:fpc-e2-planar-f} and \ref{lem:fpc-v2-planar-f}
establish the intractability. 
Suppose the second condition does not hold. The maximal class $\F$ such that $\F\not\supseteq 
\BFV_2$ and $\F\not\supseteq \BFE_2$ is $\BFL$. Lemma \ref{lem:fpc-l-all-f} states that 
for $\BFL$ counting fixed-points can be done in polynomial time. If $\XG\not\supseteq 
\Forb_\preceq(K_{3,3},K^5)$, then we know that $\XG\in\Forb_\preceq(X)$ for some planar graph $X$.
Lemmas \ref{lem:fpc-e-planar-f} and \ref{lem:fpc-v-planar-f} imply that $\#\FPE_\formula(\BFE,\XG)$ and $\#\FPE_\formula(\BFV,\XG)$ are solvable in polynomial time.
\end{proof}

\section{Conclusion}

Fixed points are an important and robust (in the sense that they exist independently of 
any update schedule) feature of discrete dynamical systems. We presented two dichotomy
theorems on the complexity of counting the number of fixed points in such a system.
Both results demonstrate that the linear boolean functions are the only function class 
such that fixed point counting is tractable independent of representations and of degrees of 
variable dependency. 

Regarding future work, it is tempting to apply our analysis framework 
(Post classes and forbidden minors) to a precise identification of
islands of predictability for more schedule-based behavioral patterns, e.g., 
gardens of Eden, predecessors, or fixed-point reachability.


\bigskip

\noindent
{\bf Acknowledgments}

\bigskip
\noindent
We thank Ernst W. Mayr (TU M\"unchen) for careful proofreading and for pointing out an error in 
an earlier version of this paper.

{\small

}

\end{document}